\documentclass{sig-alternate}

\usepackage[utf8]{inputenc}
\usepackage[T1]{fontenc}
\usepackage{graphicx}
\usepackage{amssymb}
\usepackage{amsmath}
\usepackage{mysymbols}
\usepackage{hyperref}
\usepackage{tikz}
\usepackage{dblfloatfix}
\usepackage{placeins}

\newtheorem{definition}{Definition}
\newtheorem{theorem}[definition]{Theorem}
\newtheorem{lemma}[definition]{Lemma}
\newtheorem{corollary}[definition]{Corollary}
\newtheorem{proposition}[definition]{Proposition}

\usepackage{algorithm}
\usepackage{algorithmic}

\usepackage{xcolor}

\newcommand{\Cyc}{\Phi}  
\newcommand{\bb}{\mathbf{B}}  
\newcommand{\uu}{\mathbf{U}}  

\newdef{remark}{Remark}

\begin{document}

\title{Fast algorithms for {\huge$\ell$}-adic towers over finite fields}
\numberofauthors{3}
\author{
  \alignauthor Luca De Feo\\
  \affaddr{Laboratoire PRiSM}\\
  \affaddr{Universit\'e de Versailles}\\
  \email{luca.de-feo@uvsq.fr}
  \alignauthor Javad Doliskani\\
  \affaddr{Computer Science Department}\\
  \affaddr{Western University}\\
  \email{jdoliska@uwo.ca}
  \alignauthor \'Eric Schost\\
  \affaddr{Computer Science Department}\\
  \affaddr{Western University}\\
  \email{eschost@uwo.ca}
}

\maketitle
\begin{abstract}
  Inspired by previous work of Shoup, Lenstra-De Smit and
  Couveignes-Lercier, we give fast algorithms to compute in (the first
  levels of) the $\ell$-adic closure of a finite field. In many cases,
  our algorithms have quasi-linear complexity.
\end{abstract}
\category{F.2.1}{Theory of computation}{Analysis of algorithms and problem complexity}[Computations in finite fields]
\category{G.4}{Mathematics of computing}{Mathematical software}
\terms{Algorithms,Theory}
\keywords{Finite fields, irreducible polynomials, extension towers, algebraic tori, Pell's equation, elliptic curves.}


\section{Introduction}
\label{sec:intro}

Building arbitrary finite extensions of finite fields is a fundamental
task in any computer algebra system. For this, an especially powerful
system is the ``compatibly embedded finite fields'' implemented in
Magma~\cite{MAGMA,bosma+cannon+steel97}, capable of building
extensions of any finite field and keeping track of the embeddings
between the fields.

The system described in~\cite{bosma+cannon+steel97} uses linear
algebra to describe the embeddings of finite fields. From a complexity
point of view, this is far from optimal: one may hope to compute and
apply the morphisms in quasi-linear time in the degree of the
extension, but this is usually out of reach of linear algebra
techniques. Even worse, the quadratic memory requirements make the
system unsuitable for embeddings of large degree extensions. Although
the Magma core has evolved since the publication of the paper,
experiments in Section~\ref{sec:impl} show that embeddings of large
extension fields are still out of reach.

In this paper, we discuss an approach based on polynomial arithmetic,
rather than linear algebra, with much better performance. We consider
here one aspect of the question, $\ell$-adic towers; we expect that
this will play an important role towards a complete solution.

\begin{table*}[!t]
$$
\begin{array}{c|cccc}
  \text{Condition} & \text{Initialization} & Q_i, T_i & \text{Lift, push}\\
  \hline \hline
  q = 1 \bmod \ell & O_e(\log(q))  & O(\ell^i) & O(\ell^i) \\
  q = -1 \bmod \ell & O_e(\log(q)) & O(\ell^i) & O(\Mult(\ell^i) \log(\ell^i)) \\
  - & O_e (\ell^2+\Mult(\ell) \log(q)) & O(\Mult(\ell^{i+1})\Mult(\ell)\log(\ell^i)^2) & O( \Mult(\ell^{i+1})\Mult(\ell)\log(\ell^i))\\
  4\ell \le q^{1/4} & O\tilde{_e}(\ell\log^5(q)+\ell^3) \text{~(bit)}  & O_e(\ell^2+\Mult(\ell)\log(\ell
  q)+\Mult(\ell^i)\log(\ell^i)) & O(\Mult(\ell^i) \log(\ell^i)) \\
  4\ell \le q^{1/4} & O\tilde{_e}(\ell\log^5(q))  \text{~(bit)}  + O_e(\Mult(\ell)\sqrt{q}\log(q)) & O_e(\log(q) + \Mult(\ell^i) \log(\ell^i)) & O(\Mult(\ell^i) \log(\ell^i)) \vspace{-0.5cm}
\end{array}
$$
\label{table:main}
\caption{Summary of results}
\end{table*}

Let $q$ be a power of a prime $p$, let $\F_q$ be the finite field with
$q$ elements and let $\ell$ be a prime. Our main interest in this
paper is on the algorithmic aspects of the \emph{$\ell$-adic closure}
of $\F_q$, which is defined as follows. Fix arbitrary embeddings
\begin{equation*}
  \F_q \subset \F_{q^\ell} \subset \F_{q^{\ell^2}} \subset \cdots;
\end{equation*}
then, the $\ell$-adic closure of $\F_q$ is the infinite field defined as
\begin{equation*}
  \F_q^{(\ell)} = \bigcup_{i\ge 0}\F_{q^{\ell^i}}.
\end{equation*}
We also call an \emph{$\ell$-adic tower} the sequence of extensions
$\F_q,\F_{q^\ell},\dots$ In particular, they allow us to build the
algebraic closure $\bar{\F}_q$ of $\F_q$, as there is an isomorphism
\begin{equation}
  \label{eq:tensor}
  \bar{\F}_q \isom \bigotimes_{\ell\text{ prime}} \F_q^{(\ell)},
\end{equation}
where the tensor products are over $\F_q$; we will briefly mention
below the algorithmic counterpart of this equality.

We present here algorithms that allow us to ``compute'' in the
first levels of $\ell$-adic towers (in a sense defined hereafter); at
level $i$, our goal is to be able to perform all basic operations in
quasi-linear time in the extension degree $\ell^i$.  We do not discuss
the representation of the base field $\F_q$, and we count 
operations $\{+,-,\times,\div\}$ in $\F_q$ at unit cost.


The techniques we use are inspired by those in~\cite{df+schost12},
which dealt with the Artin-Schreier case $\ell=p$ (see
also~\cite{DoSc12}, which reused these ideas in the case $\ell=2$): we
construct families of irreducible polynomials with special properties,
then give algorithms that exploit the special form of those
polynomials to apply the embeddings. Because they are treated in the
references~\cite{df+schost12,DoSc12}, {\em we exclude the cases $\ell=p$
and $\ell=2$}.

The field $\F_{q^{\ell^i}}$ will be represented as $\F_q[X_i]/\langle
Q_i\rangle$, for some irreducible polynomial $Q_i \in
\F_q[X_i]$. Letting $x_i$ be the residue class of $X_i$ modulo $Q_i$
endows $\F_{q^{\ell^i}}$ with the monomial basis
\begin{equation}
  \label{eq:uni-basis1}
  \uu_i = (1,x_{i},x_{i}^2,\ldots,x_{i}^{\ell^{i}-1}).
\end{equation}
Let $\Mult : \N \rightarrow
\N$ be such that polynomials in $\F_q[X]$ of degree less than $n$ can
be multiplied in $\Mult(n)$ operations in $\F_q$, under the
assumptions of~\cite[Ch.~8.3]{vzGG}; using FFT multiplication, one can
take $\Mult(n)\in O(n\log (n) \log\log (n))$. Then, multiplications and
inversions in $\F_q[X_i]/\langle Q_i \rangle$ can be done in
respectively $O(\Mult(\ell^i))$ and $O(\Mult(\ell^i)\log(\ell^i))$
operations in $\F_q$~\cite[Ch.~9-11]{vzGG}. This is almost optimal, as
both results are quasi-linear in $[\F_{q^{\ell^i}}:\F_q]=\ell^i$.

Computing embeddings requires more work. For this problem, it is
enough consider a pair of consecutive levels in the tower, as any
other embedding can be done by applying repeatedly this elementary
operation. Following again~\cite{df+schost12}, we introduce two
slightly more general operations, {\em lift} and {\em push}.

To motivate them, remark that for $i \ge 2$, $\F_{q^{\ell^{i}}}$ has
two natural bases as a vector space over $\F_q$. The first one is via
the monomial basis $\uu_{i}$ seen above, corresponding to the
univariate model $\F_q[X_{i}]/ \langle Q_{i} \rangle$. The second one
amounts to seeing $\F_{q^{\ell^{i}}}$ as a degree $\ell$ extension of
$\F_{q^{\ell^{i-1}}}$, that is, as
\begin{equation}\label{eq:QiTi}
\F_q[X_{i-1},X_i]/\langle Q_{i-1}(X_{i-1}), T_i(X_{i-1},X_i)\rangle,  
\end{equation}
for some polynomial $T_i$ monic of degree $\ell$ in $X_{i}$, and of
degree less than $\ell^{i-1}$ in $X_{i-1}$.  The corresponding basis is
bivariate and involves $x_{i-1}$ and $x_i$:
\begin{equation}
  \label{eq:bi-basis}
  \bb_{i} = (1,\ldots,x_{i-1}^{\ell^{i-1}-1},\ldots,x_i^{\ell-1},\ldots,x_{i-1}^{\ell^{i-1}-1}x_i^{\ell-1}).
\end{equation}
{\em Lifting} corresponds to the change of basis from $\bb_i$ to
$\uu_i$; {\em pushing} is the inverse transformation.

Lift and push allow us to perform embeddings as a particular case, but
they are also the key to many further operations. We do not give
details here, but we refer the reader
to~\cite{df+schost12,DoSc12,LeSc12} for examples such as the
computation of relative traces, norms or characteristic polynomials,
and applications to solving Artin-Schreier or quadratic equations,
given in~\cite{df+schost12} and~\cite{DoSc12} for respectively
$\ell=p$ and $\ell=2$.

Table~\ref{table:main} summarizes our main results.  Under various
assumptions, it gives costs (counted in terms of operations in $\F_q$)
for initializing the construction, building the polynomials $Q_i$ and
$T_i$ from~Eq.\eqref{eq:QiTi}, and performing lift and push. $O_e(\ )$
indicates probabilistic algorithms with expected running time, and
$O\tilde{_e}(\ )$ indicates the additional omission of logarithmic
factors.  Two entries mention bit complexity, as they use an elliptic
curve point counting algorithm.

In all cases, our results are close to being linear-time in $\ell^i$,
up to sometimes the loss of a factor polynomial in $\ell$.  Except for
the (very simple) case where $q=1 \bmod \ell$, these results are new,
to the best of our knowledge. To otbain them, we use two
constructions: the first one (Section~\ref{sec:LDtower}) uses
cyclotomy and descent algorithms; the second one
(Section~\ref{sec:fibers}) relies on the construction of a sequence of
fibers of isogenies between algebraic groups. 

These constructions are inspired by previous work due to respectively
Shoup~\cite{Shoup90,shoup94} and Lenstra / De
Smit~\cite{lenstra+desmit08-stdmodels}, and Couveignes /
Lercier~\cite{couveignes+lercier11}. We briefly discuss them here and
give more details in the further sections.

Lenstra and De Smit~\cite{lenstra+desmit08-stdmodels} address a
question similar to ours, the construction of the $\ell$-adic closure
of $\F_q$ (and of its algebraic closure), with the purpose of
standardizing it. The resulting algorithms run in
polynomial time, but (implicitly) rely on linear algebra and
multiplication tables, so quasi-linear time is not directly reachable.
References~\cite{Shoup90,shoup94,couveignes+lercier11} discuss a
related problem, the construction of irreducible polynomials over
$\F_q$; the question of computing embeddings is not considered.  Note
that the results in~\cite{couveignes+lercier11} are {\em
  quasi-linear}; they rely however on an algorithm by Kedlaya and
Umans~\cite{KeUm11} that works only in a boolean model, and as a
result share this specificity.

To conclude the introduction, let us mention a few applications of our
results. A variety of computations in number theory and algebraic
geometry require constructing new extension fields and moving elements
from one to the other. As it turns out, in many cases, the $\ell$-adic
constructions considered here are sufficient: two examples
are~\cite{df10,GaSc12}, both in relation to torsion subgroups of
Jacobians of curves.

The main question remains of course the cost of computing in arbitrary
extensions. As showed by Eq.~\eqref{eq:tensor}, this boils down to the
study of $\ell$-adic towers, as done in this paper, together with
algorithms for computing in \emph{composita}.
References~\cite{Shoup90,shoup94,couveignes+lercier11} deal with
related questions for the problem of computing irreducible polynomials;
a natural follow-up to the present work is to study the cost of
embeddings and similar changes of bases in this more general context.


\section{Quasi-cyclotomic towers}
\label{sec:LDtower}

In this section, we discuss a construction of the $\ell$-adic tower
over $\F_q$ inspired by previous work of Shoup~\cite{Shoup90,shoup94},
Lenstra-De Smit~\cite{lenstra+desmit08-stdmodels} and
Couveignes-Lercier~\cite{couveignes+lercier11}. The results of this
section establish rows 1 and 3 of Table~\ref{table:main}.

The construction starts by building an extension $\K_0 =
\F_q[Y_0]/\langle P_0 \rangle$, such that the residue class $y_0$ of
$Y_0$ is a non $\ell$-adic residue in $\K_0$ (we discuss this in more
detail in the first subsection); we let $r$ be the degree of
$P_0$.

By~\cite[Th.~VI.9.1]{lang}, for $i\ge 1$, the polynomial
$Y_i^{\ell^i}-y_0 \in \K_0[Y_i]$ is irreducible, so that $\K_i=
\K_0[Y_i]/\langle Y_i^{\ell^i}-y_0\rangle$ is a field with $q^{r
  \ell^i}$ elements.  If we let $y_i$ be the residue class of $Y_i$ in
$\K_i$, these fields are naturally embedded in one another by the
isomorphism $\K_{i+1} \simeq \K_i[Y_{i+1}]/\langle
Y_{i+1}^\ell-y_i\rangle$; in particular, the relation
$y_{i+1}^\ell=y_i$ holds.

In order to build $\F_{q^{\ell^i}}$, we apply a descent process, for
which we follow an idea of Shoup's. For $i \ge 0$, let $x_i$ be the
trace of $y_i$ over a subfield of index $r$:
\begin{equation}\label{eq-def:xi}
x_i = \sum_{j = 0}^{r-1} y_i^{q^{\ell^i j}}.  
\end{equation}
Then,~\cite[Th.~2.1]{Shoup90} proves that $\F_q(x_i)=\F_{q^{\ell^i}}$
(see Figure~\ref{fig:ladic}). In particular,
the minimal polynomials of $x_1,x_2,\dots$ over $\F_q$ are
the irreducible polynomials $Q_i$ we are interested in.

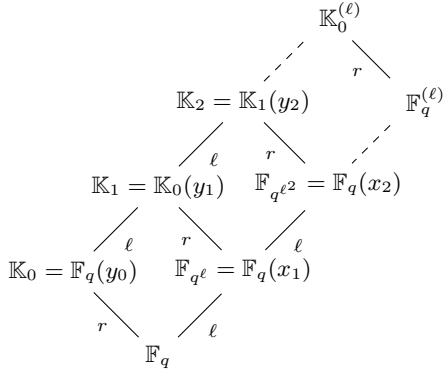
\begin{figure}[h]
  \centering
  \begin{tikzpicture}[node distance=1.6cm]
    \node(Q){$\F_q$};
    \node(Q0)[above left of=Q]{$\K_0=\F_q(y_0)$};
    \node(K1)[above right of=Q]{$\F_{q^\ell}=\F_q(x_1)$};
    \node(Q1)[above right of=Q0]{$\K_1=\K_0(y_1)$};
    \node(K2)[above right of=K1]{$\F_{q^{\ell^2}}=\F_q(x_2)$};
    \node(Q2)[above right of=Q1]{$\K_2=\K_1(y_2)$};
    \node(Koo)[above right of=K2]{$\quad \F_q^{(\ell)}$};
    \node(Qoo)[above right of=Q2]{$\quad\K_0^{(\ell)}$};
    \draw (Q) edge node[auto]{\scriptsize$r$} (Q0)
              edge node[auto,swap]{\scriptsize$\ell$} (K1)
          (K1) edge node[auto]{\scriptsize$r$} (Q1)
               edge node[auto,swap]{\scriptsize$\ell$} (K2)
          (Q1) edge node[auto]{\scriptsize$\ell$} (Q0)
          (Q2) edge node[auto,swap]{\scriptsize$r$} (K2)
               edge node[auto]{\scriptsize$\ell$} (Q1)
               edge[dashed] (Qoo)
          (Koo) edge[dashed] (K2)
               edge node[auto]{\scriptsize$r$} (Qoo);
  \end{tikzpicture}
  \caption{The $\ell$-adic towers over $\F_q$ and $\K_0$.}
  \label{fig:ladic}
\end{figure}

We will show here how to compute these polynomials, the polynomials
$T_i$ introduced in Eq.~\eqref{eq:QiTi} and how to perform lift and
push.  To this effect, we will define more general minimal
polynomials: for $0 \le j \le i$, we will let $Q_{i,j} \in
\F_q(x_j)[X_i]$ be the minimal polynomial of $x_i$ over $\F_q(x_j)$,
so that $Q_{i,j}$ has degree $\ell^{i-j}$, with in particular
$Q_{i,0}=Q_i$ and $Q_{i,i-1}=T_i(x_{i-1},X_i)$.

In Subsections~\ref{ssec:T1} and~\ref{ssec:T2}, we discuss favorable
cases, where $\ell$ divides respectively $q-1$ and $q+1$. The first
case is folklore; it yields the fastest and simplest algorithms; our
results for the second case are close to, but distinct from, previous
work of Gurak~\cite{gurak06} -- we will revisit these cases in
Section~\ref{sec:fibers} and account for their naming convention. Our
results in the general case (Subsection~\ref{ssec:gal}) are slower,
but still quasi-linear in $\ell^i$, up to a factor polynomial in
$\ell$.

Shoup used this setup to compute $Q_i$ in time quad\-ratic in
$\ell^i$~\cite[Th.~11]{shoup94}. It is noted there that using {\em
  modular composition} techniques~\cite[Ch.~12]{vzGG}, this behavior
could be improved to get a subquadratic exponent in $\ell^i$, up to an
extra cost polynomial in $\ell$.  For $\ell=3$ (where we are in one
the first two cases), Couveignes and Lercier make a similar remark
in~\cite[\S~2.4]{couveignes+lercier11}; using a result by Kedlaya and
Umans~\cite{KeUm11} for modular composition, they derive for any
$\varepsilon > 0$ a cost of $3^{i(1+\varepsilon)}O(\log(q))$ {\em bit}
operations, up to polynomial terms in $\log\log(q)$.

In this section, and in the rest of this paper, if $L/K$ is a field
extension, we write $\Tr_{L/K}$, $\Norm_{L/K}$ and $\Gal_{L/K}$ for
the trace, norm and Galois group of the extension. Recall also that
the notation $O_e(\ )$ indicates an {\em expected} running time.


\subsection{Finding $P_0$}

To determine $P_0$, we compute the $\ell$-th cyclotomic polynomial
$\Cyc_\ell \in \Z[X_0]$ and factor it over $\F_q[X_0]$:
by~\cite[Th.~9]{shoup94}, this takes $O_e(\Mult(\ell)\log(\ell q))$
operations in $\F_q$.

Over $\F_q[X_0]$, $\Cyc_\ell$ splits into irreducible factors of the
same degree $r$, where $r$ is the order of $q$ in $\Z/\ell\Z$ (so $r$
divides $\ell-1$); let $F_0$ be one of these factors. By construction,
there exist non $\ell$-adic residues in $\F_q[X_0]/\langle F_0
\rangle$. Once such a non-residue $y_0$ is found, we simply let $P_0$
be its minimal polynomial over $\F_q$ (which still has degree $r$);
given $y_0$, computing $P_0$ takes $O(r^2)$ operations in $\F_q$.

Following~\cite{Shoup90,shoup94,couveignes+lercier11}, we pick $y_0$
at random: we expect to find a non-residue after $O(1)$ trials;
by~\cite[Lemma~15]{shoup94}, each takes
$O_e(\Mult(\ell)\log(r)+\Mult(r)\log(\ell)\log(r)+\Mult(r)\log(q))$
operations in $\F_q$. An alternative due to Lenstra and De Smit is to
take iterated $\ell$-th roots of $X_0 \bmod F_0$ until we find a
non-residue: this idea is helpful in making the construction
canonical, but more costly, so we do not consider it.


\subsection{$T_1$-type extensions}\label{ssec:T1}

We consider here the simplest case, where $\ell$ divides $q-1$; the
(classical) facts below give the first row of Table~\ref{table:main}.

In this case, $\Phi_\ell$ splits into linear factors over $\F_q$ (so
$r=1$). The polynomial $P_0$ is of the form $Y_0-y_0$, where $y_0$ is
a non $\ell$-adic residue in $\F_q$; since we can bypass the
factorization of $\Phi_\ell$, the cost of initialization is
$O_e(\log(q))$ operations in $\F_q$. Besides, no descent is required:
for $i \ge 0$, we have $\K_i=\F_{q^{\ell^i}}$ and $x_i=y_i$; the
families of polynomials we obtain are
\begin{equation}
  \label{eq:T1}
  Q_i=X_i^{\ell^i}-y_0 \quad\text{and}\quad T_i=X_{i}^\ell-X_{i-1}.
\end{equation}
Lifting amounts to taking $F = \sum_{0 \le j < \ell^{i+1}} f_j
x_{i+1}^j$ and rewriting it as a bivariate polynomial in
$x_i,x_{i+1}$, using the rule
$$x_{i+1}^j = x_i^{\,j {\rm~div~} \ell} x_{i+1}^{\,j \bmod \ell}.$$
Pushing does the converse operation, using the rule
$$x_i^e x_{i+1}^f = x_{i+1}^{e \ell + f}.$$ Both 
use only exponent arithmetic, and no operation in~$\F_q$. 


\subsection{$T_2$-type extensions}
\label{ssec:T2}

Next, we consider the case where $\ell$ divides $q+1$, so that
$\Phi_\ell$ splits into quadratic factors over $\F_q$ (that is,
$r=2$). We also require that $y_0$ has norm $1$ over $\F_q$ (see below
for a discussion); we can then deduce an expression for the
polynomials $Q_{i,j} \in \F_q(x_{j})[X_i]$.

\begin{proposition}
  \label{th:T2-resultant}
  For $1 \le j < i$, $Q_{i,j}$ satisfies
  \begin{equation}
    \label{eq:T2-relpols}
    Q_{i,j}(X_i) = Y^{\ell^{i-j}} + Y^{-\ell^{i-j}} - x_j \mod Y^2-X_iY+1.
  \end{equation}
\end{proposition}
\begin{proof}
  Since $\Norm_{\K_0/\F_q}(y_0)=1$, $\Norm_{\K_i/\F_q(x_i)}(y_i)$ is
  an $\ell^i$-th root of unity. But $\ell$ does not divide $q-1$, so
  $1$ is the only such root in $\F_q$, and by induction on $i$ it also
  is the only root in $\F_q(x_i)$; hence, the minimal polynomial of
  $y_i$ over $\F_q(x_i)$ is $Y_i^2 -x_i Y_i +1$. By composition, it
  follows that the minimal polynomial of $y_i$ over $\F_q(x_{j})$ is
  $Y_i^{2\ell^{i-j}}-x_{j} Y_i^{\ell^{i-j}}+1$. Taking a resultant to
  eliminate $Y_i$
  between these two polynomials gives the following relation between
  $x_{j}$ and $x_i$:
  \begin{equation*}
    Q_{i,j}(X_i)^2 = \mathrm{Res}_{Y_i}(Y_i^{2\ell^{i-j}}-x_{j}Y_i^{\ell^{i-j}}+1,\; Y_i^2-X_i Y_i+1).
  \end{equation*}
  By direct calculation, this is equivalent to
  Eq.~\eqref{eq:T2-relpols}.
\end{proof}

This proposition would allow us to compute $Q_{i,j}$ in time
$O(\Mult(\ell^{i-j}))$ by repeated squaring. In
Section~\ref{ssec:fibers-T2}, we use arithmetic geometry to give a
better algorithm, and to efficiently find a $y_0$ satisfying the
hypotheses; we leave the algorithms for lift and push to
Section~\ref{sec:lift-push}.


\subsection{The general case}\label{ssec:gal}

Finally, we discuss the general situation, where make no
assumption on the behavior of $\Phi_\ell$ in $\F_q[X]$. This completes
the third row of Table~\ref{table:main}, using the bound $r\in O(\ell)$.

Because $r=[\K_0 : \F_q]$ divides $\ell-1$, it is coprime with
$\ell$. Thus, $Q_i$ remains the minimal polynomial of $x_i$ over
$\K_0$, and more generally $Q_{i,j}$ remains the minimal polynomial of
$x_i$ over $\K_{j}$; this will allow us to replace $\F_q$ by $\K_0$ as
our base field. We will measure all costs by counting operations in
$\K_0$, and we will deduce the cost over $\F_q$ by adding a factor
$O(\Mult(r)\log(r))$ to account for the cost of arithmetic in $\K_0$.

For $i \ge 0$, since $\K_i=\K_0[Y_i]/\langle Y_i^{\ell^i}-y_0\rangle$,
we represent the elements of $\K_i$ on the basis $\{y_i^e \mid 0 \le e
< \ell^i\}$; for instance, $x_i$ is written on this basis as
\begin{equation}\label{eq-def:xiyi}
x_i = \sum_{j = 0}^{r-1} y_i^{q^{\ell^i j} \bmod \ell^i}
y_0^{q^{\ell^i j} {\rm~div~} \ell^i}. 
\end{equation}
Our strategy is to convert between two univariate bases of $\K_i$,
$\{y_i^e \mid 0 \le e < \ell^i\}$ and $\{x_i^e \mid 0 \le e <
\ell^i\}$. In other words, we show how to apply the isomorphism
$$\Psi_i: \K_i=\K_0[Y_i]/\langle Y_i^{\ell^i}-y_0\rangle \to
\K_0[X_i]/\langle Q_{i,0}\rangle$$ and its inverse; we will compute
the required polynomials $Q_{i,0}$ and $Q_{i,i-1}$ as a byproduct. In
a second time, we will use $\Psi_i$ to perform push and lift between
the monomial basis in $x_i$ and the bivariate basis in
$(x_{i-1},x_i)$.

We will factor $\Psi_i$ into elementary
isomorphisms
$$\Psi_{i,j}: \K_j[X_i]/\langle Q_{i,j}\rangle \to
\K_{j-1}[X_i]/\langle Q_{i,j-1}\rangle, \quad j=i,\dots,1.$$ To start
the process, with $j=i$, we let $Q_{i,i}=X_i-x_i \in \K_i[X_i]$, so
that $\K_i=\K_i[X_i]/\langle Q_{i,i} \rangle$.
Take now $j \le i$ and suppose that $Q_{i,j}$ is known. We are going to
factor $\Psi_{i,j}$ further as $\Phi''_{i,j} \circ \Phi'_{i,j} \circ
\Phi_{i,j}$, by introducing first the isomorphism
$$\varphi_j: \K_j \to \K_{j-1}[Y_j]/\langle Y_j^\ell-y_{j-1}\rangle.$$
The forward direction is a push from the monomial basis in $y_j$ to
the bivariate basis in $(y_{j-1},y_j)$ and the inverse is a lift; none
of them involves any arithmetic operation (see
Subsection~\ref{ssec:T1}).  Then, we deduce the isomorphism
$$\Phi_{i,j}: \K_j[X_i]/\langle Q_{i,j} \rangle \to
\K_{j-1}[Y_j,X_i]/\langle Y_j^\ell-y_{j-1}, Q^\star_{i,j}\rangle,$$
where $Q^\star_{i,j}$ is obtained by applying $\varphi_j$ to all
coefficients of $Q_{i,j}$. Since $\Phi_{i,j}$ consists in a
coefficient-wise application of $\varphi_j$, applying it or its
inverse costs no arithmetic operations.

Next, changing the order of $Y_j$ and $X_i$, we deduce that there
exists $S_{i,j}$ in $\K_{j-1}[X_j]$ and an isomorphism
\begin{multline*}
\Phi'_{i,j}: \K_{j-1}[Y_j,X_i]/\langle Y_j^\ell-y_{j-1}, Q^\star_{i,j}\rangle
\to\\ \K_{j-1}[X_i, Y_j]/\langle Q_{i,j-1}, Y_j-S_{i,j}\rangle,
\end{multline*}
where $\deg(Q^\star_{i,j},X_i)=\ell^{i-j}$ and
$\deg(Q_{i,j-1},X_i)=\ell^{i-j+1}$. 
\begin{lemma}
  From $Q^\star_{i,j}$, we can compute $Q_{i,j-1}$ and $S_{i,j}$ in
  $O(\Mult(\ell^{i+1})\log(\ell^i))$ operations in $\K_0$.  Once this
  is done, we can apply $\Phi'_{i,j}$ or its inverse in
  $O(\Mult(\ell^{i+1}))$ operations in~$\K_0$.
\end{lemma}
\begin{proof}
  We obtain $Q_{i,j-1}$ and $S_{i,j}$ from the resultant and degree-1
  subresultant of $Y_j^\ell-y_{j-1}$ and $Q^\star_{i,j}$ with respect
  to $Y_j$, computed over the polynomial ring $\K_{j-1}[X_i]$. This is
  done by the algorithms of~\cite{Reischert97,LiRo01}, using
  $O(\Mult(\ell^{i+1})\log(\ell))$ operations in $\K_0$ (for this
  analysis, and all others in this proof, we assume that we use Kronecker's
  substitution for multiplications). To obtain $S_{i,j}$, we
  invert the leading coefficient of the degree-1 subresultant modulo
  the resultant $Q_{i,j-1}$; this takes
  $O(\Mult(\ell^{i})\log(\ell^i))$ operations in $\K_0$.

  Applying $\Phi'_{i,j}$ amounts to taking a polynomial $A(Y_j,X_i)$ 
  reduced modulo $\langle Y_j^\ell-y_{j-1}, Q^\star_{i,j}\rangle$
  and reducing it modulo $\langle Q_{i,j-1}, Y_j-S_{i,j}\rangle$. This
  is done by computing $A(S_{i,j},X_i)$, doing all operations
  modulo $Q_{i,j-1}$. Using Horner's scheme, this takes $O(\ell)$ 
  operations $(+,\times)$ in $\K_{j-1}[X_i]/\langle Q_{i,j-1}\rangle$,
  so the complexity claim follows.

  Conversely, we start from $A(X_i)$ reduced modulo $Q_{i,j-1}$; we
  have to reduce it modulo $\langle Y_j^\ell-y_{j-1},
  Q^\star_{i,j}\rangle$. This is done using the fast Euclidean
  division algorithm with coefficients in $\K_{j-1}[Y_j]/\langle
  Y_j^\ell-y_{j-1}\rangle$ for $O(\Mult(\ell^{i+1}))$
  operations in $\K_0$.
\end{proof}

The last isomorphism $\Phi''_{i,j}$ is trivial:
$$\Phi''_{i,j}: \K_{j-1}[X_i, Y_j]/\langle Q_{i,j-1}, Y_j-S_{i,j}\rangle
\to \K_{j-1}[X_i]/\langle Q_{i,j-1}\rangle$$
forgets the variable $Y_j$; it requires no arithmetic operation.

Taking $j=i,\dots,1$ allows us to compute $Q_{i,i-1}$ and $Q_{i,0}$
for $O(i^2\Mult(\ell^{i+1})\log(\ell))$ operations in $\K_0$. Composing
the maps $\Psi_{i,j}$, we deduce further that we can apply $\Psi_i$ or
its inverse for $O(i\Mult(\ell^{i+1}))$ operations in $\K_0$.  

We claim that we can then perform push and lift between the monomial
basis in $x_i$ and the bivariate basis in $(x_{i-1},x_i)$ for the same
cost. Let us for instance explain how to lift.

We start from $A$ written on the bivariate basis in $(x_{i-1},x_i)$;
that is, $A$ is in $\K_0[X_{i-1},X_i]/\langle Q_{i-1},
T_{i}\rangle$. Apply $\Psi_{i-1}$ to its coefficients in
$x_i^0,\dots,x_i^{\ell-1}$, to rewrite $A$ as an element of
$$\K_0[Y_{i-1},X_i]/\langle Y_{i-1}^{\ell^{i-1}}-y_{i-2},
T_{i}\rangle = \K_{i-1}[X_i]/\langle Q_{i,i-1} \rangle.$$ Applying
$\Psi_{i,i}^{-1}$ gives us the image of $A$ in $\K_i$, and applying
$\Psi_i$ finally brings it to $\K_0[X_i]/\langle Q_{i}\rangle$.


\section{Towers from irreducible fibers}
\label{sec:fibers}

In this section we discuss another construction of the $\ell$-adic
tower based on work of Couveignes and
Lercier~\cite{couveignes+lercier11}. The results of this section are
summarized in rows 2, 4 and 5 of Table~\ref{table:main}. This
construction is not unrelated to the ones of the previous section, and
indeed we will start by showing how those of Sections~\ref{ssec:T1}
and~\ref{ssec:T2} reduce to it.

Here is the bottom line of Couveignes' and Lercier's idea. Let $G, G'$
be integral algebraic $\F_q$-groups of the same dimension and let $\phi:
G' \rightarrow G$ be a surjective, separable algebraic group morphism.
Let $\ell$ be the degree of $\phi$; then, the set of points $x \in G$
with fiber $G'_x$ of cardinality $\ell$ is a nonempty open subset $U
\subset G$. If the induced homomorphism $G'(\F_q) \rightarrow G(\F_q)$
of groups is not surjective then there are points of $G(\F_q)$ with
fibers lying in algebraic extensions of $\F_q$. Assume that we are
able to choose $\phi$ so that we can find one of these points contained
in $U$, with an irreducible fiber, and apply a linear projection to
this fiber (e.g., onto an axis). The resulting polynomial is
irreducible of degree dividing $\ell$ (and expectedly equal to
$\ell$). If we can repeat the construction with a new map $\phi':G''\to
G'$, and so on, the sequence of extensions makes an $\ell$-adic tower
over $\F_q$.


\subsection{Towers from algebraic tori}
\label{ssec:fibers-T2}
In~\cite{couveignes+lercier11}, Couveignes and Lercier explain how
their idea yields the tower of Section~\ref{ssec:T1}. Consider the
multiplicative group $\mathbb{G}_m$: this is an algebraic group of
dimension one, and $\mathbb{G}_m(\F_q)$ has cardinality $q-1$.  The
$\ell$-th power map defined by $\phi:X\mapsto X^\ell$ is a degree
$\ell$ algebraic endomorphism of $\mathbb{G}_m$, surjective over the
algebraic closure.

Suppose that $\ell$ divides $q-1$, and let $\eta$ be a non $\ell$-adic
residue in $\F_q$ ($\eta$ plays here the same role as $y_0$ in
Section~\ref{sec:LDtower}). For any $i>0$, the fiber $\phi^{-i}(\eta)$
is defined by $X^{\ell^i}-\eta$: we recover the construction of
Subsection~\ref{ssec:T1}.

More generally, let $\F_{q^n}/\F_q$ be a finite extension and define
its \emph{maximal torus} as
\begin{equation}
  \label{eq:Tn}
  T_n = \{\alpha\in\F_{q^n} \;|\; \Norm_{\F_{q^n}/\F_{q^m}}(\alpha) = 1 
  \text{ for any $m|n$} \}.
\end{equation}
$T_n$ is a multiplicative subgroup of $\F_q^\ast$, and, by Weil
descent, an algebraic group over $\F_q$. It has dimension $\euler(n)$,
cardinality $\Phi_n(q)$, and is isomorphic to
$\mathbb{G}_m^{\euler(n)}$ over $\bar \F_q$~\cite{rubin+silverberg03,voskresenskii98}.

We now detail how the construction of Section~\ref{ssec:T2} can be
obtained by considering the torus $T_2$; this will allow us to start
completing the second row in Table~\ref{table:main}.

\begin{lemma}
  Let $\Delta\in\F_q$ be a quadratic non-residue if $p\ne2$, or such
  that $\Tr_{\F_q/\F_2}(\Delta)=1$ otherwise. Let $\delta=\sqrt{\Delta}$
  or $\delta^2+\delta=\Delta$ accordingly. The
  maximal torus $T_2$ of $\F_q(\delta)/\F_q$ is isomorphic to the
  \emph{Pell conic}
  \begin{equation}
    \label{eq:Pell}
    C \;:\; 
    \begin{cases}
      x^2 - \Delta y^2 = 4 &\text{if $p\ne2$,}\\
      x^2\Delta + xy + y^2 = 1 &\text{if $p=2$.}
    \end{cases}
  \end{equation}
  Multiplication in $T_2$ induces a group law on $C$. The neutral
  element is $(2,0)$ if $p\ne2$, or $(0,1)$ if $p=2$. The sum of two
  points $P=(x_1,y_1)$ and $Q=(x_2,y_2)$ is defined
  by
  \begin{equation*}
    P\oplus Q =
    \begin{cases}
      \displaystyle
      \left(\frac{x_1x_2 + \Delta y_1y_2}{2},\; \frac{x_1y_2 + x_2y_1}{2}\right) &
      \text{if $p\ne2$,}\\
      \left(x_1x_2 + x_1y_2 + x_2y_1,\; x_1x_2\Delta + y_1y_2\right) &
      \text{if $p=2$.}
    \end{cases}
  \end{equation*}
\end{lemma}
\begin{proof}
  The isomorphism follows by Weil descent with respect to the basis
  $(1/2,\delta/2)$ if $p\ne2$, or $(\delta,1)$ if $p=2$. Indeed, by
  virtue of Eq.~\eqref{eq:Tn}, an element $(x,y)$ of $\F_q(\delta)$
  belongs to $T_2$ if and only if its norm over $\F_q$ is $1$.

  Let $\sigma$ be the generator of $\Gal_{\F_q(\delta)/\F_q}$. For
  $p=2$, clearly $\delta^\sigma=-\delta$. For $p\ne2$, by
  Artin-Schreier theory,
  $\Tr_{\F_q(\delta)/\F_q}(\delta)=\Tr_{\F_q/\F_2}(\Delta)=1$, hence
  $\delta^\sigma=1+\delta$. In both cases, Eq.~\eqref{eq:Pell}
  follows.  The group law is obtained by direct calculation.
\end{proof}

Pell conics are a classic topic in number theory\cite{lenstra02-pell}
and computer science, with applications to primality proving,
factorization~\cite{lemmermeyer03,hambleton12} and
cryptography~\cite{rubin-silverberg+crypto03}. 

As customary, we denote by $[n](x,y)$ the $n$-th scalar multiple of a
point $(x,y)$.

\begin{lemma}
  \label{th:T2-divpol}
  If $(n,p)=1$, then $[n]$ is a separable endomorphism of $C$ of
  degree $n$, given by the rational maps
  \begin{equation}
    \label{eq:Pell-rec}
    [n](x,y) = 
    \begin{cases}
      \bigr(P_n(x), y R_n(x)\bigl) & \text{if $p\ne2$,}\\
      \bigr(P_n(x), y R_n(x) + R_{n-1}(x)\bigl) & \text{if $p=2$.}      
    \end{cases}
  \end{equation}
  where $P_n$ and $R_n$ are defined by the initial values
  \begin{align*}
    P_0 &= 2,\qquad P_1=X,\\
    R_0 &= 0,\qquad R_1=1,
  \end{align*}
  and by the same recurrence $u_{n+1} = Xu_{n} - u_{n-1}$.
\end{lemma}
\begin{proof}
  We know that $C\isom\mathbb{G}_m$, thus $C[n]\isom\Z/n\Z$ and $[n]$
  is separable of degree $n$. Eq.~\eqref{eq:Pell-rec} is shown by
  induction using Eq.~\eqref{eq:Pell} and the group law.
\end{proof}


\begin{theorem}
  \label{th:T2-irred}
  Let $\eta\in\F_q(\delta)$ be a non $\ell$-adic residue in $T_2$, and
  let $P=(\alpha,\beta)$ be its image in $C/\F_q$. For any $i>0$, the
  polynomials $P_{\ell^i}-\alpha$ are irreducible. Their roots are the
  abscissas of the images in $C/\F_{q^{\ell^i}}$ of the $\ell^i$-th
  roots of $\eta$.
\end{theorem}
\begin{proof}
  By~\cite[Th.~VI.9.1]{lang}, the polynomial $X^{\ell^i}-\eta$ is
  irreducible. Its roots correspond to the fiber $[\ell^i]^{-1}(P)$,
  and the Galois group of $\F_{q^{\ell^i}}/\F_q$ acts transitively on
  them.

  Two points of $C$ have the same abscissa if and only if they are
  opposite. But $\eta\ne\eta^{-1}$, hence all the points in
  $[\ell^i]^{-1}(P)$ have distinct abscissa.  By
  Lemma~\ref{th:T2-divpol}, $P_{\ell^i}-\alpha$ vanishes precisely on
  those abscissas and is thus irreducible.
\end{proof}

We can now apply our results to the computation of the polynomials
$Q_i$ and $T_i$ of Section~\ref{ssec:T2}.

\begin{corollary}
  The polynomials $Q_{i,j}$ of Prop.~\ref{th:T2-resultant}
  satisfy
  \begin{equation*}
    Q_{i,j}(X_i) = P_{\ell^{i-j}}(X_i) - x_j.
  \end{equation*}
\end{corollary}
\begin{proof}
  We have already shown that $\Norm_{\K_j/\F_q(x_j)}(y_j)=1$ for any
  $j$, thus $y_j$ is a non $\ell$-adic residue in
  $T_2/\F_q(x_j)$. Independently of the characteristic and of the
  element $\Delta\in\F_q(x_j)$ chosen, the abscissa of the image of
  $y_j$ in $C/\F_q(x_j)$ is $\Tr_{\K_j/\F_q(x_j)}y_j=x_j$. The
  statement follows from the previous theorem.
\end{proof}

\begin{corollary}
  The polynomials $Q_{i,j}$ can be computed using $O(\ell^{i-j})$
  operations.
\end{corollary}
\begin{proof}
  From the previous corollary, it is enough to compute $P_n$ using
  $O(n)$ operations. We write $P_n = \sum_i c_{n,i}X^{n-i}$, from
  Lemma~\ref{th:T2-divpol} we deduce that
  \begin{equation}
    c_{n,i} = c_{n-1,i} - c_{n-2,i-2}.
  \end{equation}
  By induction, it is immediate that $c_{n,i}=0$ for $i$ odd, and that
  signs alternate for $i$ even, so we remove the odd coefficients and
  take absolute values. The new coefficients $b_{n,k}=\lvert
  c_{n+k,2k}\rvert$ satisfy the relation
  \begin{equation*}
    b_{n,k} = b_{n-1,k} + b_{n-1,k-1},
  \end{equation*}
  which is the same as Pascal's relation; we actually obtain the
  $(1,2)$-Pascal triangle, also called Lucas'
  triangle~\cite{benjamin10}.  In the same way, we can prove that the
  even coefficients of $R_n$ are the entries of Pascal's triangle with
  alternating signs.
  
  As is well-known, the coefficients of
  Lucas' triangle are related to those of Pascal's by
  \begin{equation}
    \label{eq:lucas-tri}
    b_{n,k} = \binom{n}{k} + \binom{n-1}{k-1} = \frac{n+k}{n}\binom{n}{k}.
  \end{equation}
  Using Eq.~\eqref{eq:lucas-tri} and the sign alternation property,
  we get
  \begin{equation*}
    \frac{c_{n,2k+2}}{c_{n,2k}} = 
    -\frac{(n-2k)(n-2k-1)}{(n-k-1)(k+1)}.
  \end{equation*}
  The last equation gives the formula to compute all the coefficients
  of $P_n$ using $O(n)$ operations in $\F_p$. Indeed, since we know
  the $c_{n,2k}$'s are the image mod $p$ of integers, we compute them
  using multiplications and divisions in $\Q_p$ with relative
  precision 1.
\end{proof}

We are left with the problem of finding the non $\ell$-adic residue
$\eta$ to initialize the tower. As before, this will be done by random
sampling and testing.

\begin{lemma}
  \label{th:montgomery}
  Let $P=(\alpha,\beta)$ be a point on $C$. For any $n$, there is a
  formula to compute the abscissa of $[\pm n]P$, using $O(\log n)$
  operations in $\F_q$, and not involving $\beta$.
\end{lemma}
\begin{proof}
  Observe that if $n=2$, the abscissa of $[\pm 2]P$ is $\alpha^2-2$
  (for any $p$).  Let $P'=(\alpha',\beta')$, and let $\gamma$ be the
  abscissa of $P\ominus P'$. By direct computation we find that the
  abscissa of $P\oplus P'$ is $\alpha\alpha'-\gamma$ (for any $p$);
  this formula is called a \emph{differential addition}.  Thus, $O(1)$
  operations are needed for a doubling or a differential addition. To
  compute the abscissa of $[\pm n]P$, we use the ladder algorithm
  of~\cite{montgomery}, requiring $O(\log n)$ doublings and
  differential additions.
\end{proof}

\begin{proposition}
  The abscissa of a point $P\in C/\F_q$ satisfying the
  conditions of Theorem~\ref{th:T2-irred} can be found using $O_e(\log
  q)$ operations in $\F_q$.
\end{proposition}
\begin{proof}
  We randomly select $\alpha\in\F_q$ and test that it belongs to
  $C$. If $p\ne2$, this amounts to testing that $\alpha^2-4$ is a
  quadratic non-residue in $\F_q$, a task that can be accomplished
  with $O(\log q)$ operations. If $p=2$, by Artin-Schreier theory this
  is equivalent to $\Tr_{\F_q/\F_2}(1/\alpha^2)=1$, which can be
  tested in $O(\log q)$ operations in $\F_q$.

  Then we check that $P$ is a non $\ell$-adic residue by verifying
  that $[(q+1)/\ell]P$ is not the group identity. By
  Lemma~\ref{th:montgomery}, this computation requires $O(\log q)$
  operations.
  About half of the points of $\F_q$ are quadratic non-residues, and
  about $1-1/\ell$ of them are the abscissas of points with the
  required order, thus we expect to find the required element after
  $O_e(1)$ trials.
\end{proof}

It is natural to ask whether a similar construction could be applied
to any $\ell$. If $r$ is the order of $q$ modulo $\ell$, the natural
object to look at is $T_r$, but here we are faced with two
problems. First, multiplication by $\ell$ is now a degree
$\ell^{\euler(r)}$ map, thus its fibers have too many points; instead,
isogenies of degree $\ell$ should be considered. Second, it is an open
question whether $T_r$ can be parameterized using $\euler(r)$
coordinates; but even assuming it can be, we are still faced with the
computation of a univariate annihilating polynomial for a set embedded
in a $\euler(r)$-dimensional space, a problem not known to be feasible
in quasi-linear time. Studying this generalization is another natural
follow-up to the present work.


\subsection{Towers from elliptic curves}
\label{sec:elliptic}

Since it seems hard to deal with higher dimensional algebraic tori, it
is interesting to look at other algebraic groups. Being
one-dimensional, elliptic curves are good candidates. In this
section, we quickly review Couveignes' and Lercier's construction,
referring to~\cite{couveignes+lercier11} for details, and point out
the modifications needed in order to build towers (as opposed to
constructing irreducible polynomials).

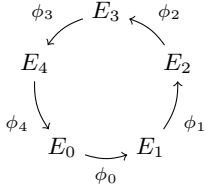
\begin{figure}
  \centering
  \begin{tikzpicture}
    \def\n{4}
    \foreach \i in {0,...,\n} {
      \pgfmathparse{360/(\n+1)*(\i-1/2) - 90}
      \let\angle\pgfmathresult
      \draw (\angle:1) node (E\i) {$E_\i$};
    }
    \foreach \i in {0,...,\n} {
      \pgfmathparse{int(mod(\i+1, \n+1))}
      \let\j\pgfmathresult
      \draw (E\i) edge[->,bend right=18] node[auto,swap] {\scriptsize$\phi_\i$} (E\j);
    }
  \end{tikzpicture}
  \caption{The isogeny cycle of $E_0$.}
  \label{fig:volcano}
\end{figure}

Let $\ell$ be a prime different from $p$ and not dividing $q-1$. Let
$E_0$ be an elliptic curve whose cardinality is a multiple of
$\ell$. By Hasse's bound, this is only possible if $\ell\le q +
2\sqrt{q} + 1$.  An \emph{isogeny} is an algebraic group morphism
between two elliptic curves that is surjective in the algebraic
closure. It is said to be rational over $\F_q$ if it is invariant
under the $q$-th power map; such an isogeny exists if and only if the
curves have the same number of points over $\F_q$. An isogeny of
degree $n$ is separable if and only if $n$ is prime to $p$, in which
case its kernel contains exactly $n$ points.  Because of the
assumptions on $\ell$, there exists an $e\ge1$ such that, for any curve
$E$ isogenous to $E_0$, the $\F_q$-rational part of $E[\ell]$ is
cyclic of order $\ell^e$.

Suppose for simplicity, that $p\ne2,3$ and let $E_0$ be expressed as
the locus
\begin{equation}
  E_0 \;:\; y^2 = x^3 + ax + b,
  \quad\text{with $a,b\in\F_q$},
\end{equation}
plus one point at infinity.  We denote by $H_0$ the unique subgroup of
$E_0/\F_q$ of order $\ell$, and by $\phi_0$ the unique isogeny whose
kernel is $H_0$; we then label $E_1$ the image curve of $\phi_0$. We
go on denoting by $H_i$ the unique subgroup of $E_i/\F_q$ of order
$\ell$, and by $\phi_i:E_i\to E_{i+1}$ the unique isogeny with kernel
$H_i$. The construction is depicted in Figure~\ref{fig:volcano}.

\begin{lemma}
  \label{th:class-number}
  Let $E_0,E_1,\dots$ be defined as above, there exists $n\in
  O(\sqrt{q}\log (q))$ such that $E_n$ is isomorphic to $E_0$.
\end{lemma}
\begin{proof}
  It is shown in \cite[\S~4]{couveignes+lercier11} that the
  isogenies $\phi_i$ are \emph{horizontal} in the sense
  of~\cite{kohel}, hence they necessarily form a cycle. Let $t$ be the
  trace of $E_0$, the length of the cycle is bounded by the class
  number of $\Q[X]/(X^2-tX-q)$, thus by Minkowski's bound it is in
  $O(\sqrt{q}\log (q))$.
\end{proof}

In what follows, the index $i$ is to be understood modulo the length of
the cycle. This is a slight abuse, because $E_n$ is isomorphic but not
equal to $E_0$, but it does not hide any theoretical or computational
difficulty.

Under the former assumptions, it is proved
in~\cite[\S~4]{couveignes+lercier11} that if $P$ is a point of
$E_i$ of order divisible by $\ell^e$, if
$\psi=\phi_{i-1}\circ\phi_{i-2}\circ\cdots\circ\phi_{j}$, then the
fiber $\psi^{-1}(P)$ is irreducible and has cardinality $\ell^{i-j}$.
Knowing $E_i$, Vélu's formulas~\cite{velu71} allow us to express the
isogenies $\phi_i$ as rational fractions
\begin{equation}
  \begin{aligned}
    \phi_i: E_i &\to E_{i+1},\\
    (x,y) &\mapsto \left(\frac{f_i(x)}{g_i(x)}, y\left(\frac{f_i(x)}{g_i(x)}\right)'\right),
  \end{aligned}
\end{equation}
where $g_i$ is the square polynomial of degree $\ell-1$ vanishing on
the abscissas of the affine points of $H_i$, and $f_i$ is a polynomial
of degree $\ell$. 

There is a subtle difference between our setting and Couveignes' and
Lercier's. The goal of~\cite{couveignes+lercier11} is to compute an
extension of degree $\ell^i$ of $\F_q$ for a fixed $i$: this can be
done by going forward $i$ times, then taking the fiber of a point of $E_i$ by
the isogenies $\phi_{i-1}, \ldots, \phi_0$. In our case, we are
interested in building extensions of degree $\ell^i$
\emph{incrementally}, i.e.\ without any \emph{a priori} bound on
$i$. Thus, we have to walk \emph{backwards} in the isogeny cycle: if
$\eta\in\F_q$ is the abscissa of a point of $E_0$ of order $\ell^e\ne
2$, we will use the following polynomials to define the $\ell$-adic
tower:
\begin{align*}
  T_1 &= f_{-1}(X_1) - \eta g_{-1}(X_1),\\ 
  T_i &= f_{-i}(X_i) - X_{i-1} g_{-i}(X_i).
\end{align*}

The following theorem gives the time for building the tower; lift and
push are detailed in the next section. 

\begin{theorem}\label{theo:elliptic}
  \sloppy
  Suppose $4\ell\le q^{1/4}$, and under the above assumption.
  Initializing the $\ell$-adic tower requires
  $O\tilde{_e}(\ell\log^5(q)+\ell^3)$ bit operations; and building the
  $i$-th level requires $O_e(\ell^2+\Mult(\ell)\log(\ell
  q)+\Mult(\ell^i)\log(\ell^i))$ operations in $\F_q$.
\end{theorem}
\begin{proof}
  For the initialization, \cite[\S~4.3]{couveignes+lercier11}
  shows that if $4\ell\le q^{1/4}$, a curve $E_0$ with the required
  number of points can be found in $O\tilde{_e}(\ell\log^5(q))$ bit
  operations. We also need to compute the $\ell$th modular polynomial
  $\Phi_\ell\bmod p$; for this, we compute it over $\Z$ with
  $\tildO(\ell^3)$ bit operations~\cite{enge09}, then reduce it
  modulo $p$.

  To build the $i$-th level, we first need to find the equation of
  $E_{-i}$. For this, we evaluate $\Phi_\ell$ at $j(E_{-i+1})$, using
  $O(\ell^2)$ operations. The resulting polynomial has two roots in
  $\F_q$, namely $j(E_{-i})$ and $j(E_{-i+2})$. We factor it using
  $O_e(\Mult(\ell)\log(\ell q))$ operations~\cite[Ch~14]{vzGG}. Once
  $E_{-i}$ is known, we find an $\ell$-torsion point using $O_e(\log
  q)$ operations, and apply Vélu's formulas to compute $\phi_{-i}$. We
  deduce the polynomial $T_i$, and $Q_i$ is obtained using
  $O(\Mult(\ell^i)\log(\ell^i))$ operations using
  Algorithm~\ref{alg:compose} given in the next section.
\end{proof}

\begin{remark}
  \label{rk:elliptic}
  Instead of computing the cycle step by step, we could compute it
  entirely during the initialization phase, by using Vélu's formulas
  alone to compute $E_1,E_2,\dots$ until we hit $E_0$ again. By doing
  so, we avoid using the modular polynomial $\Phi_\ell$ at each new
  level. By Lemma~\ref{th:class-number}, this requires
  $O_e(\ell\sqrt{q}\log(q))$ operations. This is not asymptotically good
  in $q$, but for practical values of $q$ and $\ell$ the cycle is often
  small and this approach works well. This is accounted for in the last
  row of Table~\ref{table:main}.  
\end{remark}


\section{Lifting and pushing}
\label{sec:lift-push}

The previous constructions of $\ell$-adic towers based on irreducible
fibers share a common structure that allows us to treat lifting and
pushing in a unified way. Renaming the variables $(X_{i-1},X_i)$ as
$(X,Y)$, the polynomials $(Q_{i-1},Q_i,T_i)$ as $(R,S,T)$, the
extension at level $i$ is described as
$$\F_q[Y]/ \langle S(Y) \rangle \quad\text{and}\quad \F_q[X,Y]/\langle
R(X), T(X,Y) \rangle,$$ with $R$ of degree $\ell^{i-1}$, $S$ of degree
$\ell^i$, and where $T(X,Y)$ has the form $f(Y)-X g(Y)$, with $\deg(f)
=\ell$, $\deg(g) < \ell$ and $\gcd(f,g)=1$; possibly, $g=1$. In all
this section, $f$, $g$ and their degree $\ell$ are fixed.

Lift is the conversion from the bivariate basis associated to the
right-hand side to the univariate basis associated to the left-hand side;
push is the inverse. Using the special shape of the polynomial $T$,
they reduce to composition and decomposition of rational functions, as
we show next. These results fill in all missing entries in the lift / push 
column of Table~\ref{table:main}.


\subsection{Lifting}

Let $P$ be in $\F_q[X,Y]$ and $n$ be in $\N$, with $\deg(P,X)< n$. We
define $P[f,g,n]$ as
$$P[f,g,n] = g^{n-1} P\left (\frac fg, Y\right) \in \F_q[X,Y].$$ If
$P=\sum_{i=0}^{n-1} p_i(Y) X^i$, then $P[f,g,n] = \sum_{i=0}^{n-1}
p_if^ig^{n-1-i}$.  We first give an algorithm to compute this
expression, then show how to relate it to lifting; when $g=1$,
Algorithm~\ref{alg:compose} reduces to a well known algorithm for
polynomial composition~\cite[Ex.~9.20]{vzGG}.

\begin{algorithm}[t]
  \caption{Compose}
  \label{alg:compose}
  \begin{algorithmic}[1]
    \REQUIRE $P\in \F_q[X,Y]$, $f,g\in \F_q[Y]$, $n\in\N$
    \IF {$n = 1$} 
    \RETURN $P$
    \ELSE
    \STATE $m \la \lceil n/2\rceil$
    \STATE Let $P_0,P_1$ be such that $P = P_0 + X^mP_1$
    \STATE $Q_0 \la$ Compose($P_0, f, g, m$)
    \STATE $Q_1 \la$ Compose($P_1, f, g, n-m$)
    \STATE $Q \la Q_0g^{n-m} + Q_1f^m$  \label{alg:compose:res}
    \RETURN $Q$
    \ENDIF
  \end{algorithmic}
\end{algorithm}

\begin{theorem}
  \label{th:compose}
  On input $P,f,g,n$, with $\deg(P,X)<n$ and $\deg(P,Y) < \ell$,
  Algorithm~\ref{alg:compose} computes $Q=P[f,g,n]$ using $O(\Mult(\ell
  n)\log(n))$ operations in $\F_q$.
\end{theorem}
\begin{proof}
  If $n=1$, the theorem is obvious. Suppose $n>1$, then $P_0$ and
  $P_1$ have degrees less than $m$ and $n-m$ respectively. By
  induction hypothesis,
  \begin{equation*}
    \begin{aligned}
      Q_0 &= P_0[f,g,m] = \sum_{i=0}^{m-1}p_if^ig^{m-1-i},\\
      Q_1 &= P_1[f,g,n-m] = \sum_{i=0}^{n-m-1}p_{i+m}f^ig^{n-m-1-i}.   
    \end{aligned}
  \end{equation*}
  Hence,
  \begin{equation*}
    Q = \sum_{i=0}^{m-1}p_if^ig^{n-1-i} +
    \sum_{i=0}^{n-m-1}p_{i+m}f^{i+m}g^{n-m-1-i} =
    P[f,g,n].
  \end{equation*}
  The only step that requires a computation is
  Step~\ref{alg:compose:res}, costing $O(\Mult(\ell n))$ operations in
  $\F_q$. The recursion has depth $\log(n)$, hence the overall
  complexity is $O(\Mult(\ell n)\log(n))$.
\end{proof}

\begin{corollary}
  At level $i$, one can perform the lift operation using
  $O(\Mult(\ell^i)\log(\ell^i))$ operations in $\F_q$.
\end{corollary}
\begin{proof}
  We start from an element $\alpha$ written on the bivariate basis, that
  is, represented as $A(X,Y)$ with $\deg(A,X)<n=\ell^{i-1}$ and
  $\deg(A,Y)<\ell$ (note that $\ell n =\ell^i$).  We compute the
  univariate polynomials $A^\star=A[f,g,n]$ and $\gamma=g^{n-1}$ using
  $O(\Mult(\ell^i)\log(\ell^i))$ operations in $\F_q$; then the lift
  of $\alpha$ is $A^\star/\gamma$ modulo $S$. The inverse of $\gamma$
  is computed using $O(\Mult(\ell n)\log(\ell n))$ operations, and the
  multiplication adds an extra $O(\Mult(\ell n))$.
\end{proof}


\subsection{Pushing}

We first deal with the inverse of the question dealt with in
Theorem~\ref{th:compose}: starting from $Q \in \F_q[Y]$, reconstruct
$P \in \F_q[X,Y]$ such that $Q=P[f,g,n]$. When $g=1$,
Algorithm~\ref{alg:decompose} reduces to Algorithm~9.14
of~\cite{vzGG}.

\begin{algorithm}[t]
  \caption{Decompose}
  \label{alg:decompose}
  \begin{algorithmic}[1]
    \REQUIRE $Q,f,g,h\in \F_q[Y]$, $n \in \N$ 
    \IF {$n=1$} 
    \RETURN $Q$
    \ELSE 
    \STATE $m \la \lceil n/2 \rceil$ 
    \STATE $u \la 1/ g^{n-m}\bmod f^m$ \label{alg:decompose:xgcd} 
    \STATE $Q_0 \la Q u \bmod f^m$ 
    \STATE $Q_1 \la (Q-Q_0 g^{n-m}) {\rm~div~} f^m$ 
    \STATE $P_0 \la$ Decompose($Q_0, f, g, h, m$) 
    \STATE $P_1 \la$ Decompose($Q_1, f, g, h, n-m$) 
    \RETURN $P_0 + X^m P_1$ 
    \ENDIF
  \end{algorithmic}
\end{algorithm}

\begin{theorem}
  On input $Q,f,g,h,n$, with $\deg(Q) < \ell n$ and $h = 1/g \bmod f$,
  Algorithm~\ref{alg:decompose} computes a polynomial $P\in \F_q[X,Y]$
  such that $\deg(P,X)<n$, $\deg(P,Y) <\ell$ and $Q=P[f,g,n]$ using
  $O(\Mult(\ell n)\log(n))$ operations in $\F_q$.
\end{theorem}
\begin{proof}
  We prove the theorem by induction. If $n=1$, the statement is
  obvious, so let $n> 1$. The polynomials $Q_0$ and $Q_1$ verify $Q =
  Q_0g^{n-m} + Q_1f^m.$ By construction, $Q_0$ has degree less than $\ell
  m$. Since $\deg(g) < \ell$, this implies that $Q_0 g^{n-m}$ has
  degree less than $\ell n$; thus, $Q_1$ has degree less than $\ell
  (n-m)$. By induction, $P_0$ and $P_1$ have degree less than $m$,
  resp. $n-m$, in $X$, and less than  $\ell$ in $Y$, and
  \begin{equation*}
    \begin{aligned}
      Q_0 &= P_0[f,g,m] = \sum_{i=0}^{m-1} p_{0,i}f^ig^{m-1-i},\\
      Q_1 &= P_1[f,g,n-m] = \sum_{i=0}^{n-m-1} p_{1,i}f^ig^{n-m-1-i}.
    \end{aligned}
  \end{equation*}
  Hence, $P=P_0+X^mP_1$ has degree less than $n$ in $X$ and less than $\ell$ 
  in $Y$, and the following proves correctness:
  \begin{multline*}
    P[f,g,n] = \sum_{i=0}^{m-1}p_{0,i}f^ig^{n-1-i} + 
    \sum_{i=m}^{n-1}p_{1,i-m}f^ig^{n-1-i} \\
    =P_0[f,g,m]g^{n-m} + P_1[f,g,n-m]f^m = Q.
  \end{multline*}
  At Step~\ref{alg:decompose:xgcd}, we do as follows: starting from
  $h=1/g \bmod f$, we deduce $1/g^{n-m} \bmod f$ in time
  $O(\Mult(\ell)\log(n))$ by binary powering mod $f$. We also compute
  $g^{n-m}$ in time $O(\Mult(\ell n))$ by binary powering, and we use
  Newton iteration (starting from $1/g^{n-m} \bmod f$) to deduce
  $1/g^{n-m} \bmod f^m$ in time $O(\Mult(\ell n))$. All other steps
  cost $O(\Mult(\ell n))$; the recursion has depth $\log(n)$,
  so the total cost is $O(\Mult(\ell n)\log(n))$.
\end{proof}

\begin{corollary}
  At level $i$, one can perform the push operation using
  $O(\Mult(\ell^i)\log(\ell^i))$ operations in $\F_q$.
\end{corollary}
\begin{proof}
  Given $\alpha$ represented by a univariate polynomial $A(Y)$ of
  degree less than $\ell n$, with $n =\ell^{i-1}$. We compute
  $g^{n-1}$ and $A^\star = g^{n-1} A \bmod S$ using $O(\Mult(\ell^i))$
  operations. Then, we compute $h=1/g \bmod f$ in time
  $O(\Mult(\ell)\log(\ell))$ and apply Algorithm~\ref{alg:decompose}
  to $A^\star$, $f$, $g$, $h$ and $n$. The result is a bivariate
  polynomial $B$, representing $\alpha$ on the bivariate basis. The
  dominant phase is Algorithm~\ref{alg:decompose}, costing
  $O(\Mult(\ell^i)\log(\ell^i))$ operations in $\F_q$.
\end{proof}


\section{Implementation}
\label{sec:impl}
To demonstrate the interest of our constructions, we made a very basic
implementation of the towers of Sections~\ref{ssec:fibers-T2}
and~\ref{sec:elliptic} in Sage~\cite{Sage}. It relies on Sage's
default implementation of quotient rings of $\F_p[X]$, which itself
uses NTL \cite{shoup2003ntl} for $p=2$ and FLINT~\cite{hart2010flint}
for other primes. Towers based on elliptic curves are constructed
using the algorithm described in Remark~\ref{rk:elliptic}. The source
code is available on De Feo's web page.

We compare our implementation against three ways of constructing
$\ell$-adic towers in Magma:
\begin{itemize}
\item We construct the levels from bottom to top using the default
  finite field constructor \verb+GF()+. For the parameters we were
  able to test, Magma uses tables of precomputed Conway polynomials
  and automatically computes embeddings on
  creation.\footnote{See \url{http://magma.maths.usyd.edu.au/magma/releasenotes/2/14}}
\item We construct the highest level of the tower first, then all the
  lower levels using the \verb+sub<>+ constructor.
\item We construct the levels from bottom to top using random dense
  polynomials, then we call the \verb+Embed()+ function.  We do not
  account for the time spent finding the irreducible polynomials.
\end{itemize}

\begin{figure}
  \centering
  \includegraphics[height=4.7cm]{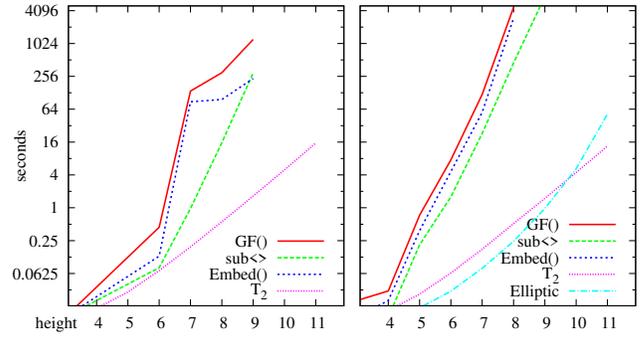}
  \caption{Times for building $3$-adic towers on top of $\F_2$ (left)
    and $\F_5$ (right), in Magma (first three lines) and using our
    code.}
  \label{fig:build}
\end{figure}

We ran tests on an Intel Xeon E5620 clocked at 2.4 GHz, using Sage 5.5
and Magma 2.18.12. The time required for the creation of $3$-adic
towers of increasing height is summarized in
Figure~\ref{fig:build}; the timings of our algorithms are labeled 
T$_2$ and Elliptic. Computations that took more than 4GB RAM were
interrupted.

Despite its simplicity, our code consistently outperforms Magma on
creation time. On the other hand, lift and push operations take
essentially no time in Magma, while in all the tests of
Figure~\ref{fig:build} we measured a running time almost perfectly
linear for one push followed by one lift, taking approximately $70\mu
s$ per coefficient (this is in the order of a second around level
10). Nevertheless, the large gain in creation time makes the
difference in lift and push tiny, and we are convinced that an
optimized C implementation of the algorithms of
Section~\ref{sec:lift-push} would match Magma's performances.


\bigskip\noindent \textbf{Aknowledgements.}  De Feo would like to thank Antoine
Joux and J\'er\^ome Pl\^ut for fruitful discussions. Schost acknowledges
support from NSERC and the CRC program.

\scriptsize
\bibliographystyle{abbrv}
\bibliography{defeo}

\begin{thebibliography}{10}

\bibitem{benjamin10}
A.~T. Benjamin.
\newblock The {L}ucas triangle recounted.
\newblock In {\em Congressus Numerantum, Proceedings of the Twelfth Conference
  on Fibonacci Numbers and their Applications}, volume 200, pages 169--177,
  2010.

\bibitem{MAGMA}
W.~Bosma, J.~Cannon, and C.~Playoust.
\newblock The {MAGMA} algebra system {I}: the user language.
\newblock {\em Journal of Symbolic Computation}, 24(3-4):235--265, 1997.

\bibitem{bosma+cannon+steel97}
W.~Bosma, J.~Cannon, and A.~Steel.
\newblock Lattices of compatibly embedded finite fields.
\newblock {\em Journal of Symbolic Computation}, 24(3-4):351--369, 1997.

\bibitem{couveignes+lercier11}
J.-M. Couveignes and R.~Lercier.
\newblock Fast construction of irreducible polynomials over finite fields.
\newblock {\em To appear in the Israel Journal of Mathematics}, July 2011.

\bibitem{df10}
L.~De~Feo.
\newblock Fast algorithms for computing isogenies between ordinary elliptic
  curves in small characteristic.
\newblock {\em Journal of Number Theory}, 131(5):873--893, May 2011.

\bibitem{df+schost12}
L.~De~Feo and E.~Schost.
\newblock Fast arithmetics in {A}rtin-{S}chreier towers over finite fields.
\newblock {\em Journal of Symbolic Computation}, 47(7):771--792, July 2012.

\bibitem{DoSc12}
J.~Doliskani and {\'E}.~Schost.
\newblock A note on computations in degree $2^k$-extensions of finite fields,
  2012.
\newblock Manuscript.

\bibitem{enge09}
A.~Enge.
\newblock Computing modular polynomials in quasi-linear time.
\newblock {\em Mathematics of Computation}, 78(267):1809--1824, 2009.

\bibitem{GaSc12}
P.~Gaudry and E.~Schost.
\newblock Point-counting in genus 2 over prime fields.
\newblock {\em Journal of Symbolic Computation}, 47(4):368–400, 2012.

\bibitem{gurak06}
S.~Gurak.
\newblock Minimal polynomials for gauss periods with f=2.
\newblock {\em Acta Arithmetica}, 121(3):233, 2006.

\bibitem{hambleton12}
S.~A. Hambleton.
\newblock Generalized {L}ucas-{L}ehmer tests using {P}ell conics.
\newblock {\em Proceedings of the American Mathematical Society},
  140:2653--2661, 2012.

\bibitem{hart2010flint}
W.~Hart.
\newblock Fast library for number theory: an introduction.
\newblock {\em International Conference on Mathematical Software--ICMS 2010},
  pages 88--91, 2010.

\bibitem{KeUm11}
K.~S. Kedlaya and C.~Umans.
\newblock Fast polynomial factorization and modular composition.
\newblock {\em SIAM J. Computing}, 40(6):1767--1802, 2011.

\bibitem{kohel}
D.~Kohel.
\newblock {\em Endomorphism rings of elliptic curves over finite fields}.
\newblock PhD thesis, University of California at Berkley, 1996.

\bibitem{lang}
S.~Lang.
\newblock {\em Algebra}.
\newblock Springer, 3rd edition, Jan. 2002.

\bibitem{LeSc12}
R.~Lebreton and {\'E}.~Schost.
\newblock Algorithms for the universal decomposition algebra.
\newblock In {\em ISSAC'12}, pages 234--241. ACM, 2012.

\bibitem{lemmermeyer03}
F.~Lemmermeyer.
\newblock {Conics - a Poor Man's Elliptic Curves}, Nov. 2003.

\bibitem{lenstra02-pell}
H.~W. Lenstra.
\newblock Solving the {P}ell equation.
\newblock {\em Notices of the {AMS}}, 49(2):182--192, 2002.

\bibitem{lenstra+desmit08-stdmodels}
H.~W. Lenstra and B.~De~Smit.
\newblock Standard models for finite fields: the definition, 2008.

\bibitem{LiRo01}
T.~Lickteig and M.~Roy.
\newblock Sylvester–habicht sequences and fast cauchy index computation.
\newblock {\em Journal of Symbolic Computation}, 31(3):315 -- 341, 2001.

\bibitem{montgomery}
P.~L. Montgomery.
\newblock Speeding the pollard and elliptic curve methods of factorization.
\newblock {\em Mathematics of Computation}, 48(177), 1987.

\bibitem{Reischert97}
D.~Reischert.
\newblock Asymptotically fast computation of subresultants.
\newblock In {\em ISSAC}, pages 233--240. ACM, 1997.

\bibitem{rubin-silverberg+crypto03}
K.~Rubin and A.~Silverberg.
\newblock {Torus-Based} cryptography.
\newblock In D.~Boneh, editor, {\em Advances in Cryptology - CRYPTO 2003},
  volume 2729 of {\em Lecture Notes in Computer Science}, pages 349--365,
  Berlin, Heidelberg, 2003. Springer Berlin / Heidelberg.

\bibitem{rubin+silverberg03}
K.~Rubin and A.~Silverberg.
\newblock Algebraic tori in cryptography.
\newblock In {\em In High Primes and Misdemeanours: Lectures in Honour of the
  60th birthday of Hugh Cowie Williams}, volume~41 of {\em Fields Institute
  Communications}. American Mathematical Society, 2004.

\bibitem{Shoup90}
V.~Shoup.
\newblock New algorithms for finding irreducible polynomials over finite
  fields.
\newblock {\em Math. Comp.}, 54:435--447, 1990.

\bibitem{shoup94}
V.~Shoup.
\newblock Fast construction of irreducible polynomials over finite fields.
\newblock {\em Journal of Symbolic Computation}, 17(5):371--391, 1994.

\bibitem{shoup2003ntl}
V.~Shoup.
\newblock {NTL}: A library for doing number theory.
\newblock \url{http://www.shoup.net/ntl}, 2003.

\bibitem{Sage}
W.~A. Stein and Others.
\newblock {\em {S}age {M}athematics {S}oftware ({V}ersion 5.5)}.
\newblock The Sage Development Team, 2013.

\bibitem{velu71}
J.~V{\'{e}}lu.
\newblock Isog{\'{e}}nies entre courbes elliptiques.
\newblock {\em Comptes Rendus de l'Acad\'{e}mie des Sciences de Paris},
  273:238--241, 1971.

\bibitem{vzGG}
J.~von~zur Gathen and J.~Gerhard.
\newblock {\em Modern computer algebra}.
\newblock Cambridge University Press, New York, NY, USA, 1999.

\bibitem{voskresenskii98}
V.~E. Voskresenski\u{i}.
\newblock {\em Algebraic groups and their birational invariants}, volume 179.
\newblock American Mathematical Society, 1998.

\end{thebibliography}
\end{document}